\documentclass[aps,prl,reprint,amsmath,amssymb,superscriptaddress,longbibliography]{revtex4-2}

\usepackage[T1]{fontenc}
\usepackage[utf8]{inputenc}
\usepackage{mathtools}
\usepackage{dsfont}
\usepackage{xcolor}
\definecolor{linkblue}{RGB}{30,72,110}
\usepackage[colorlinks=true,linkcolor=blue,citecolor=blue,urlcolor=red]{hyperref}

\newcommand{\Tr}{\operatorname{Tr}}
\newcommand{\supp}{\operatorname{supp}}
\newcommand{\sgn}{\operatorname{sgn}}
\newcommand{\1}{\mathds 1}
\newcommand{\Sone}{\mathcal S_1}
\newcommand{\Kcal}{\mathcal K}
\newcommand{\Bcal}{\mathcal B}
\newcommand{\Hcal}{\mathcal H}
\newcommand{\Dlog}{\mathrm D\!\log}
\newcommand{\hbin}{h_2}
\newcommand{\eps}{\varepsilon}
\newcommand{\norm}[1]{\lVert #1\rVert}
\newcommand{\tno}[1]{\lVert #1\rVert_1}
\newcommand{\posproj}[1]{\1_{(0,\infty)}\!\left(#1\right)}
\newcommand{\ket}[1]{\lvert #1\rangle}

\newcommand{\proj}[1]{\lvert #1\rangle\!\langle #1\rvert}

\newtheorem{theorem}{Theorem}
\newtheorem{lemma}{Lemma}
\newtheorem{proposition}{Proposition}
\newenvironment{proof}[1][Proof]{\par\noindent\textit{#1.}\ }{\hfill$\square$\par\smallskip}

\begin{document}

\title{Sharp Quantum Entropy Mixing Rates and the Operator Layer Cake Theorem}

\author{Alexander Stottmeister}
\affiliation{Institut f\"ur Theoretische Physik, Leibniz Universit\"at Hannover, Appelstra{\ss}e 2, 30167 Hannover, Germany}

\begin{abstract}
At what rate does the von Neumann entropy of an ensemble of quantum states change under Hamiltonian evolution of its constituents? Bravyi proposed the small incremental mixing conjecture controlling the mixing rate of a binary ensemble $\{(1\!-\!p,\rho_1),(p,\rho_2)\}$ by $c\,\norm{H}\hbin(p)$ (with the binary entropy $\hbin$). The proof of Bravyi's conjecture was subsequently reduced to the matrix inequality $\tno{[A,\log B]}\!\leq\! c\,\hbin(p)$ for positive trace-class operators $A\!\leq\! B$ with $\Tr A\!=\!p$ and $\Tr B\!=\!1$ on separable Hilbert spaces by Mari\"en, Audenaert, Van Acoleyen, and Verstraete, and they conjectured the optimal constant to be $c\!=\!1$.
Here, I prove the latter conjecture using an exact integral representation of the commutator $[A,\log B]$ obtained from the operator layer cake theorem, due to Cheng and Liu. The result gives a sharp dimension-independent limit on how rapidly unitary evolution of one component can change the entropy of a binary quantum ensemble. It follows that mixing rates satisfy small incremental mixing with the optimal constant, and entangling rates of bipartite Hamiltonians are bounded by $(2\log d\!+\!1)\norm{H}$. Moreover, the result corrects a conjecture by Lieb and Vershynina for mixing rates of general ensembles.
\end{abstract}

\maketitle

\paragraph{Introduction.} The rate at which unitary dynamics generates or degrades entropy is an important dynamical quantity of quantum information theory, with applications ranging from entanglement generation by bipartite Hamiltonians \cite{dur_entanglement_2001,bennett_capacities_2003,bravyi_upper_2007} to lower bounds on decoherence times \cite{schlosshauer_decoherence_2005} and to the stability of the area law for the entanglement entropy within a gapped phase \cite{van_acoleyen_entanglement_2013,marien_entanglement_2016}.

In \cite{bravyi_upper_2007}, Bravyi studied two not obviously related problems:
For a binary ensemble $\mathcal E=\{(1-p,\rho_1),(p,\rho_2)\}$ with average state $\bar\rho=(1-p)\rho_1+p\rho_2$ whose constituent $\rho_2$ evolves under a Hamiltonian $H$, the \emph{mixing rate} is the instantaneous change $\Lambda(\mathcal E,H)$ of the von Neumann entropy of the average state. For a pure state $\ket\psi\in\Hcal_{aABb}$ evolving under a bipartite interaction $H=H_{AB}$, the \emph{entangling rate} $\Gamma(H,\psi)$ is the instantaneous change of the entanglement entropy across the cut between $aA$ and $Bb$.

Bravyi conjectured \emph{small incremental mixing} (SIM),
\begin{align}\label{eq:SIM}
  |\Lambda(\mathcal E,H)| & \!\leq\! c\,\norm{H}\,\hbin(p), &  \hbin(p) & \!=\!-p\log p\!-\!q\log q,
\end{align}
with $q\!=\!1\!-\!p$ and a universal constant $c$, showing that SIM implies Kitaev's \emph{small incremental entangling} (SIE), 
\begin{align}\label{eq:SIE}
|\Gamma(H,\psi)| & \leq c'\norm{H}\log d, & d & =\min\{\dim A,\dim B\}.
\end{align}
Related bounds on mixing rates were obtained by Lieb and Vershynina, who conjectured the analogue of \eqref{eq:SIM} with $c=1$ for general ensembles, replacing the binary entropy with the Shannon entropy \cite{lieb_upper_2013}. The SIM conjecture was first proven, in finite dimension, by Van Acoleyen, Mari\"en, and Verstraete with $c=9$ \cite{van_acoleyen_entanglement_2013}, then by Audenaert with $c=2$ via the quantum skew divergence \cite{audenaert_quantum_2014}, and finally extended to arbitrary separable Hilbert spaces by Mari\"en, Audenaert, Van Acoleyen, and Verstraete \cite{marien_entanglement_2016} with $c=11$. The latter authors isolated an underlying trace-norm inequality \cite[Thm.~4]{marien_entanglement_2016}
\begin{align}\label{eq:tin}
\tno{[A,\log B]} & \!\leq\! c\,\hbin(p),\!\! & A & \!\leq\! B,\!\! & \Tr A & \!=\! p,\!\! & \Tr B & \!=\! 1, 
\end{align}
and conjectured that $c=1$ is optimal, based on numerical evidence.

In this Letter, I provide a concise and structurally simple proof of this conjecture. Its main ingredient is the \emph{operator layer cake theorem} of Cheng and Liu \cite{cheng_error_2026}, which is equivalent to Frenkel's integral formula for the quantum relative entropy \cite{frenkel_integral_2023, cheng_operator_2025}. The operator layer cake theorem allows for the conversion of $[A,\log B]$ into an integral of commutators $[A-B,P_t]$ for suitable spectral projections of the perturbation $A-t'(B-A)$ of $A$ by $B-A$.\footnote{$t' = t\tfrac{p}{q}$ with $p = \Tr A$ and $q=1-p$.} An elementary estimate of each commutator $[A-B,P_t]$ produces, after integration, the two terms in \eqref{eq:SIM} of the binary entropy $h_{2}(p)$ with $p = \Tr A$.

\paragraph{Main result.} Given a separable Hilbert space $\Hcal$, the trace-class, compact operators, and bounded operators are denoted by $\Sone(\Hcal)$, $\Kcal(\Hcal)$, and $\Bcal(\Hcal)$ respectively. $\tno{\ \cdot\ }$, $\norm{\ \cdot\ }_2$, and $\norm{\ \cdot\ }$ denote the trace, Hilbert--Schmidt, and operator norms.

\begin{theorem}\label{thm:main}
Let $A,B\in\Sone(\Hcal)$ be positive operators on a separable Hilbert space with
\begin{align}
  0 & \leq A\leq B, &  \Tr A & = p\in[0,1], &  \Tr B & =1.
\end{align}
Let $\Hcal_0=(\ker B)^\perp$ and let $\mathcal D_0$ be the linear span of an orthonormal eigenbasis of the positive compact operator $B|_{\Hcal_0}$. Then the sesquilinear form
\begin{align}\label{eq:weakform}
  \mathfrak c(x,y) & =\langle x,A\log B\,y\rangle-\langle \log B\,x,Ay\rangle, &  x,y & \in\mathcal D_0,
\end{align}
has a unique bounded representative $[A,\log B]_{\Hcal_0}$ on $\Hcal_0$; it is trace class and satisfies the \emph{optimal} inequality
\begin{align}\label{eq:main}
  \tno{[A,\log B]_{\Hcal_0}} & \leq \hbin(p).
\end{align}
\end{theorem}

Since $0\!\leq\!\langle x,Ax\rangle\!\leq\!\langle x,Bx\rangle$, the operator $A$ satisfies $A\!=\!A|_{\Hcal_{0}}\!\oplus 0$ relative to the orthogonal decomposition $\Hcal = \Hcal_{0}\!\oplus\!\Hcal_{0}^{\perp}$. Thus, the canonical extension $[A,\log B]_{\mathrm{can}}\!:=\![A,\log B]_{\Hcal_{0}}\!\oplus 0$ leaves the trace norm unchanged. Clearly, $[A,\log B]_{\mathrm{can}}$ is simply the commutator on $\supp B$ in finite dimension. In the extremal cases $p\in\{0,1\}$, \eqref{eq:main} is tight because $\Tr A=0$ forces $A=0$, and $\Tr(B-A)=0$ forces $A=B$, and, thus, the commutator as well as $\hbin(p)$ vanish in both cases.

\paragraph{An exact commutator representation.} I first work on a finite-dimensional Hilbert space with $B>0$ (after restriction to $\supp B$) and $0<p<1$. By normalizing as
\begin{align}
  \rho & =\tfrac{1}{p}A, & \sigma & = \tfrac{1}{q}(B-A), &  B & = p\rho+q\sigma,
\end{align}
$\rho,\sigma$ are density matrices.

Let me temporarily assume that $\sigma>0$. For a self-adjoint operator $T$ over $\Hcal$, its positive and negative parts are denoted by $T_{\pm}\!:=\!\tfrac{1}{2}(|T| \pm T)$, and $\{T>S\}\!:=\!\posproj{T\!-\!S}$ denotes projection onto the positive part of a difference ($S$ being another self-adjoint operator over $\Hcal$). 
In the following, I will show the exact identity
\begin{align}\label{eq:log_com}
  [A,\log B] & = -pq\int_{0}^{\infty} dt\, \tfrac{1}{q+pt}[\sigma,P_{t}],
\end{align}
where $P_{t}:=\{\rho>t\sigma\}$.

The key ingredient is the operator layer cake representation of the operator logarithm's derivative due to Cheng and Liu \cite[Thm.~B.1]{cheng_error_2026}, which is equivalent to Frenkel's integral formula \cite[Thm.~6]{frenkel_integral_2023} as shown in \cite[Prop.~1]{cheng_operator_2025}, and has been generalized to separable Hilbert spaces recently \cite{kossmann_device-independent_2026}.
\begin{proposition}\label{prop:layercake}
For matrices $T\geq0$ and $S>0$,
\begin{align}\label{eq:layercake}
  \Dlog[S](T) & \!:=\!\tfrac{\textup{d}}{\textup{d}\epsilon}\log(S\!+\!\epsilon T)|_{\epsilon=0}\!=\!\int_{0}^{\infty} \!\!du\,\{T\!>\!uS\},
\end{align}
where $\{T>uS\}=0$ for $u\geq\norm{S^{-1/2}TS^{-1/2}}$, and $u\mapsto\{T>uS\}$ is Borel measurable in norm.
\end{proposition}

\begin{proof}[Proof of Eq.~\eqref{eq:log_com}]
Fix $s\in(0,p]$ and set $S_s=q\sigma+s\rho>0$, $T = \rho$. For $u\geq\tfrac{1}{s}$ one has $T-uS_s=(1-us)\rho-uq\sigma\leq0$ and, thus, the projection in \eqref{eq:layercake} vanishes. For $0\leq u<\tfrac{1}{s}$,
\begin{align}
  T-uS_s & = (1-us)\,(\rho-t\sigma), &  t & = \tfrac{uq}{1-us},
\end{align}
with $1-us>0$, so that $\{T>uS_s\}=\{\rho>t\sigma\}=P_{t}$. Substituting $u=\tfrac{t}{q+st}$, $du=\tfrac{q}{(q+st)^{2}}dt$, in \eqref{eq:layercake} (restricted to the integration range $[0,\tfrac{1}{s})$), results in 
\begin{align}\label{eq:kernel}
  \Dlog[q\sigma+s\rho](\rho) & = \int_{0}^{\infty} dt\,\tfrac{q}{(q+st)^2}\,P_t.
\end{align}
This representation remains valid at $s=0$ by $t=qu$.

Applying the fundamental theorem of calculus to $s\mapsto\log(q\sigma+s\rho)$ gives
\begin{align}\label{eq:log_int}
  \log B-\log(q\sigma) & =\int_{0}^{p} ds\, \Dlog[q\sigma+s\rho](\rho).
\end{align}
Combining \eqref{eq:kernel} \& \eqref{eq:log_int}, and interchanging the integrals leads to
\begin{align}\label{eq:log_rep}
  \log B-\log(q\sigma) & = p\int_{0}^{\infty} dt\, \tfrac{1}{q+pt}P_t .
\end{align}
because $\int_0^p q(q+st)^{-2}ds=p/(q+pt)$.

Finally, $[\sigma,\log(q\sigma)]\!=\!0$ and $p[\rho,\log B]\!+\!q[\sigma,\log B]\!=\![B,\log B]\!=\!0$ with $A\!=\!p\rho$, yield the desired representation \eqref{eq:log_com}.
\end{proof}

\paragraph{The entropy bound.} The integrand of \eqref{eq:log_com} admits an elementary estimate.

\begin{lemma}\label{lem:proj}
Let $\tau\in\Sone(\Hcal)$ be positive with $\Tr\tau\!=\!1$ and let $T\in\Bcal(\Hcal)$ be positive. Then,
\begin{align}\label{eq:trace_bound}
  \tno{[\tau,T]} & \leq 2\sqrt{\langle T^{2}\rangle_{\tau}-\langle T\rangle_{\tau}^{2}} \leq\norm{T},
\end{align}
where $\langle\ \cdot\ \rangle_{\tau} = \Tr(\tau\ \cdot\ )$.
\end{lemma}

Applying Lemma~\ref{lem:proj} to $\sigma$ and $P_{t}$, gives $\tno{[\sigma,P_t]}\leq1$. But, because $P_t$ is a spectral projection of $\rho-t\sigma$, it follows that $[\rho-t\sigma,P_t]=0$, i.e.\ $[\rho,P_t]=t\,[\sigma,P_t]$ for \emph{every} $t\geq0$. Now, applying Lemma~\ref{lem:proj} to $\rho$, yields $t\,\tno{[\sigma,P_t]}\leq1$. Combining both estimates implies
\begin{align}\label{eq:minbound}
  \tno{[\sigma,P_t]} & \leq\min\{1,t^{-1}\}, & t & >0.
\end{align}
Thus, applying the triangle inequality and \eqref{eq:minbound} to \eqref{eq:log_com}, proves \eqref{eq:main} under the assumption $\sigma>0$:
\begin{align}
  \tno{[A,\log B]}
  &\leq pq\int_{0}^{1}dt\, \tfrac{1}{q+pt}+pq\int_{1}^{\infty} dt\, \tfrac{1}{t(q+pt)}\notag\\
  &=-q\log q-p\log p=\hbin(p).\label{eq:entropyintegral}
\end{align}
If $\sigma$ is singular (on $\Hcal_{0}$), replace it by $\sigma_\eta=(1-\eta)\sigma+\eta\tfrac{1}{d_{B}}\1$ with $d_{B}=\dim\supp B$ and $B_\eta=p\rho+q\sigma_\eta$; then $A\leq B_\eta$, $\Tr B_\eta=1$, and $B_\eta\to B>0$ in norm with spectra in a fixed compact subset of $(0,\infty)$, so $\log B_\eta\to\log B$ in norm and \eqref{eq:main} passes to the limit. I defer the proof of the separable infinite-dimensional case ($\dim\Hcal_0=\infty$) to the Appendix.

The proof elucidates the mechanism behind the appearance of the binary entropy $\hbin(p)$. The only inequalities in the entire argument are the two applications of Lemma~\ref{lem:proj} and the triangle inequality under the integral representation \eqref{eq:log_com}. The region $t\leq1$ of the likelihood-ratio parameter contributes exactly $-q\log q$; the region $t\geq1$, where the spectral relation $[\rho,P_t]=t[\sigma,P_t]$ supplies the extra factor $t^{-1}$, contributes exactly $-p\log p$.

\paragraph{Optimality.} That the constant $c=1$ cannot be improved follows by considering a specific family of trace-class operators. For $0<\eps<2/3$ let
\begin{align}\label{eq:family} \notag
  B_{\eps} & =\begin{pmatrix}1-\frac{\eps}{2}&0\\[1mm]0&\frac{\eps}{2}\end{pmatrix}, & \ket{v_{\eps}} & =\begin{pmatrix}\sqrt{s_{\eps}}\\[1mm]\sqrt{1-s_{\eps}}\end{pmatrix}, \\
   A_{\eps} & = B_{\eps}^{\frac{1}{2}}\proj{v_{\eps}}B_{\eps}^{\frac{1}{2}},
\end{align}
with $s_\eps=\tfrac{\eps}{2(1-\eps)}\in(0,1)$. Then $0\leq A_{\eps}\!\leq\!B_{\eps}$ and $\Tr A_{\eps}=\langle v_{\eps},B_{\eps} v_{\eps}\rangle=\eps$. Since $\log B_{\eps}$ is diagonal, the commutator is off-diagonal with trace norm $2|(A_{\eps})_{12}|\log\tfrac{2-{\eps}}{\eps}$, explicitly
\begin{align}\label{eq:sharpnorm}
  \tno{[A_{\eps},\log B_{\eps}]} & 
  =\tfrac{\eps\sqrt{(2-\eps)(2-3\eps)}}{2(1-\eps)}\,
  \log\tfrac{2-\eps}{\eps}.
\end{align}
For small $\eps>0$, both \eqref{eq:sharpnorm} and $\hbin(\eps)$ are asymptotic to $\eps\log(1/\eps)$, so
\begin{align}\label{eq:ratio}
  \lim_{\eps\to0}\tfrac{\tno{[A_{\eps},\log B_{\eps}]}}{\hbin(\eps)} & =1,
\end{align}
and no constant below $1$ can hold universally. Notably, $\rho_{\eps}\!=\!\tfrac{1}{\eps}A_{\eps}$ and $\sigma_{\eps}\!=\!\tfrac{1}{1-\eps}(B_{\eps}-A_{\eps})$ are both qubit pure states.

\paragraph{Mixing and entangling rates.} Throughout this paragraph, all Hilbert spaces are finite dimensional. For the binary ensemble $\mathcal E$ from the introduction, with $\bar\rho(t)=(1-p)\rho_{1}+p\,e^{itH}\rho_{2}e^{-itH}$, the mixing rate is
\begin{align}\label{eq:rate}
  \Lambda(\mathcal E,H) & =\tfrac{\textup{d}}{\textup{d}t}S(\bar\rho(t))|_{t=0}
  =-i\Tr\!\big(H[A,\log\bar\rho]\big),
\end{align}
with $A=p\rho_{2}$ \cite[Eq.~(25)]{marien_entanglement_2016} (see also \cite{bravyi_upper_2007}), where $\log\bar\rho$ is taken on $\supp\bar\rho$, which contains $\supp\rho_{2}$ because $p\rho_{2}\leq\bar\rho$.

\begin{theorem}\label{thm:rates}
In finite dimension, the following hold:
\begin{itemize}
\item[(i)] for every binary ensemble,
\begin{align}
|\Lambda(\mathcal E,H)| & \!\leq\! \norm{H}\,\hbin(p)
\end{align}
and the right-hand side cannot be improved to $c\,\norm{H}\hbin(p)$ with $c<1$,
\item[(ii)] for every ensemble $\{(p_{i},\rho_{i})\}_{i=1}^n$ with constituents evolving under Hamiltonians $\norm{H_{i}}\leq1$,
\begin{align}\label{eq:general}
  |\Lambda|\!\leq\!\sum_{i=1}^{n}\hbin(p_{i}) & \!\leq\! \min\{h(\{p_{i}\})+1, 2\, h(\{p_{i}\})\},
\end{align}
where $h$ is the Shannon entropy. Moreover, the universal prefactor $2$ of $h$ is optimal,
\item[(iii)] for every bipartite pure state $\ket\psi\in\Hcal_{aABb}$ and interaction $H=H_{AB}$,
\begin{align}\label{eq:SIE_2}
  |\Gamma(H,\psi)|\!\leq\! d^2\hbin(\tfrac{1}{d^{2}})\,\norm{H} & \!\leq\!(2\log d\!+\!1)\norm{H},
\end{align}
where $d=\min\{\dim A,\dim B\}$.
\end{itemize}
\end{theorem}

\begin{proof}
The inequality of Item (i) follows from \eqref{eq:rate}, $|\Tr(HC)|\!\leq\!\norm{H}\tno{C}$, and Theorem~\ref{thm:main}. Optimality follows because $C=-i[A,\log\bar\rho]$ is Hermitian, so $H=\sgn(C)$ attains $|\Lambda|\!=\!\tno{[A,\log\bar\rho]}$, and every admissible pair $(A,B)$ arises from an ensemble via $\rho_{2}\!=\!\tfrac{1}{p}A$, $\rho_{1}\!=\!\tfrac{1}{q}(B-A)$ (denoted $\rho$, $\sigma$ above), and the family \eqref{eq:family} saturates the bound.

The first part of the bound in Item (ii) follows by applying (i) termwise in the rate identity $\Lambda=-i\sum_i\Tr(H_{i}[p_{i}\rho_{i},\log\bar\rho])$ for $\bar\rho(t) = \sum_{i}p_{i}e^{itH_{i}}\rho_{i}e^{-itH_{i}}$ \cite[Eq.~(13)]{marien_entanglement_2016}, since $p_{i}\rho_{i}\leq\bar\rho$, together with $-(1-x)\log(1-x)\leq x$ and $\sum_{i}p_{i}=1$. The second part follows from the Weierstraß product inequality $\prod_{j\neq i}(1\!-\!p_{j})\!\geq\!1\!-\!\sum_{j\neq i}p_{j}\!=\!p_{i}$ \cite[Eq.~3.2.27.(1)]{mitrinovic_analytic_1970}, hence $-\sum_{i}(1\!-\!p_{i})\log(1\!-\!p_{i})\!=\!-\sum_{i}\sum_{j\neq i}p_{j}\log(1\!-\!p_{i})\!=\!-\sum_{j}\sum_{i\neq j}p_{j}\log(1\!-\!p_{i})\!=\!-\sum_{j}p_{j}\log\prod_{i\neq j}(1\!-\!p_{i})\!\leq\!h(\{p_{i}\})$.
The optimality of the prefactor $2$ of $h$ follows by considering the family \eqref{eq:family} with $\rho_{2}\!=\!\rho_{\eps}$, $\rho_{1}\!=\!\sigma_{\eps}$, and $H_{1,\eps}\!=\!-H_{2,\eps}\!=\!\sgn(i[A_{\eps},\log B_{\eps}])$. This gives $|\Lambda_{\eps}|\!=\!2\tno{[A_{\eps},\log B_{\eps}]}$ and, thus, $\frac{|\Lambda_{\eps}|}{\hbin(\eps)}\to 2$.

Item (iii) is Bravyi's reduction \cite{bravyi_upper_2007} (see also \cite[Sec.~2.4]{marien_entanglement_2016}) with the SIM constant $c\!=\!1$. Specifically, assuming $\dim B\!\leq\!\dim A$, one chooses the pair $(\tfrac{1}{d^{2}}\rho_{aAB}, \rho_{aA}\!\otimes\!\tfrac{1}{d}\1_{B})$, which is admissible as a consequence of the type I analogue of the Pimsner--Popa inequality \cite[Prop.~2.1]{pimsner_entropy_1986} for the inclusion $\Bcal(\Hcal_{aA})\!\subset\!\Bcal(\Hcal_{aAB})$ together with the 
conditional expectation $E\!=\!\Tr_{B}\!\otimes\!\tfrac{1}{d}\1_{B}$ and index $d^{2}$ (see \cite[Lem.~1]{bravyi_upper_2007}).
\end{proof}

For all $d\!\geq\!2$, \eqref{eq:SIE_2} is an improvement over previously known SIE constants. In this regard, let me emphasize the following subtlety: the reduction in \cite[Eqs.~(34)--(38)]{marien_entanglement_2016} converts the \emph{one-sided} estimate $\tno{[A,\log B]}\!\leq\! c\,p\log\tfrac{1}{p}$ into the SIE constant $2c\log d$. But, the one-sided estimate with $c\!=\!1$ is \emph{false} for general admissible pairs $(A,B)$: for $\eps\!=\!\tfrac{1}{2}$ the family \eqref{eq:family} gives $\tno{[A_{\frac{1}{2}},\log B_{\frac{1}{2}}]}\!=\!\tfrac{\sqrt3}{4}\log3\!>\!\tfrac12\log2$.
The proof of Item (ii) improves the bound in \cite[Thm.~2]{marien_entanglement_2016} and shows that the conjectured bound in \cite[Conj.~3]{lieb_upper_2013} needs to be corrected by a factor of $2$, although \cite[Eq.~(8)]{lieb_upper_2013}, and the discussion surrounding it, suggest that the conjecture might have been stated erroneously.

\paragraph{Discussion and outlook.} Theorem~\ref{thm:main} settles the value of the optimal constant in Bravyi's SIM conjecture \cite{bravyi_upper_2007}. The crucial ingredient is \eqref{eq:log_com}, which represents the commutator $[A,\log B]$, using the operator layer cake theorem \cite{cheng_error_2026}, as a weighted integral of commutators between a density matrix and a related family of spectral projections. This representation may be of independent interest, e.g., for the commutator inequalities conjectured in \cite{audenaert_problems_2012} and for rate bounds beyond the von Neumann entropy \cite{vershynina_entanglement_2019}, as well as in the quantitative analysis of area-law stability \cite{van_acoleyen_entanglement_2013,marien_entanglement_2016}, where the constant of \eqref{eq:SIM} enters into the error bounds of quasi-adiabatic continuation.

The optimal constant for SIE remains open. Item (iii) of Theorem \ref{thm:rates} yields $2\log d\!+\!1$, while the conjectured value $2\log d$ \cite{marien_entanglement_2016,van_acoleyen_entanglement_2013} would follow from the one-sided estimate restricted to the special pairs $(\tfrac{1}{d^{2}}\rho_{AB},\rho_A\otimes\tfrac{1}{d}\1_B)$ produced by the reduction. Such a constrained inequality is not excluded by the counterexample above.

It is interesting to note that the proof is not intrinsically finite-dimensional. Moreover, since trace-class operators on separable Hilbert spaces correspond to normal linear functionals on type I von Neumann algebras, the results immediately ask for a generalization to arbitrary von Neumann algebras. In addition, Frenkel's integral formula has recently been extended to normal states on arbitrary von Neumann algebras \cite{luijk_sufficiency_2026,silva_integral_2026,kossmann_device-independent_2026}, where it recovers Araki's relative entropy \cite{araki_relative1_1976,araki_relative2_1977,kosaki_relative_1986} (see also \cite{ohya_quantum_1993}). Using suitable modular-theoretic analogues, this suggests that sharp mixing bounds in general von Neumann algebras and, thus, entropy-rate bounds in quantum field theory and infinite many-body systems are achievable. Together with Lauritz van Luijk and Henrik Wilming, this will be discussed elsewhere.

\begin{acknowledgments}
Specifically, I would like to thank, on the one hand, Henrik Wilming for pointing out the question concerning the optimal SIM/SIE constant to me, and, on the other hand, Lauritz van Luijk and Henrik Wilming for raising my interest in Frenkel's integral representation of the relative entropy. Moreover, I would like to thank Lauritz van Luijk, Henrik Wilming, Tobias J.~Osborne, and Frank Verstraete for commenting on an early version of this Letter.
\end{acknowledgments}

\paragraph{Assistance by large-language models (LLMs).} 
The proof idea for the main result arose from the author's interest in Frenkel's integral formula, its generalization to arbitrary von Neumann algebras, and its applications.
LLMs were used to check the validity of the proofs and to exhibit the family that proves the optimality of $c=1$. Tobias J.~Osborne provided valuable guidance on LLM verification techniques. All references in the bibliography were inserted manually. LLMs were also used for grammar checking and writing assistance.

\paragraph{On the proof.}
I obtained the proof presented here in August 2026, through various discussions with Henrik Wilming and Lauritz van Luijk in the broader context of von Neumann-algebraic quantum information theory. Henrik Wilming and I have used the problem of the optimal constant in \eqref{eq:tin} several times throughout 2025 and the first half of 2026 as a testing ground for new frontier LLMs. A discussion with and an independent LLM trial by Tobias J.~Osborne on August 19, 2026, after writing up a first draft of the proof, made it clear that the core part of the proof was reproducible by the most current models (with minimal guidance), and one might consider the result common knowledge by then. Nevertheless, I think this hopefully clear presentation of the relatively simple, and in my opinion elegant, proof justifies this manuscript.

\bibliographystyle{apsrev4-2}

\begin{thebibliography}{29}%
\makeatletter
\providecommand \@ifxundefined [1]{%
 \@ifx{#1\undefined}
}%
\providecommand \@ifnum [1]{%
 \ifnum #1\expandafter \@firstoftwo
 \else \expandafter \@secondoftwo
 \fi
}%
\providecommand \@ifx [1]{%
 \ifx #1\expandafter \@firstoftwo
 \else \expandafter \@secondoftwo
 \fi
}%
\providecommand \natexlab [1]{#1}%
\providecommand \enquote  [1]{``#1''}%
\providecommand \bibnamefont  [1]{#1}%
\providecommand \bibfnamefont [1]{#1}%
\providecommand \citenamefont [1]{#1}%
\providecommand \href@noop [0]{\@secondoftwo}%
\providecommand \href [0]{\begingroup \@sanitize@url \@href}%
\providecommand \@href[1]{\@@startlink{#1}\@@href}%
\providecommand \@@href[1]{\endgroup#1\@@endlink}%
\providecommand \@sanitize@url [0]{\catcode `\\12\catcode `\$12\catcode
  `\&12\catcode `\#12\catcode `\^12\catcode `\_12\catcode `\%12\relax}%
\providecommand \@@startlink[1]{}%
\providecommand \@@endlink[0]{}%
\providecommand \url  [0]{\begingroup\@sanitize@url \@url }%
\providecommand \@url [1]{\endgroup\@href {#1}{\urlprefix }}%
\providecommand \urlprefix  [0]{URL }%
\providecommand \Eprint [0]{\href }%
\providecommand \doibase [0]{https://doi.org/}%
\providecommand \selectlanguage [0]{\@gobble}%
\providecommand \bibinfo  [0]{\@secondoftwo}%
\providecommand \bibfield  [0]{\@secondoftwo}%
\providecommand \translation [1]{[#1]}%
\providecommand \BibitemOpen [0]{}%
\providecommand \bibitemStop [0]{}%
\providecommand \bibitemNoStop [0]{.\EOS\space}%
\providecommand \EOS [0]{\spacefactor3000\relax}%
\providecommand \BibitemShut  [1]{\csname bibitem#1\endcsname}%
\let\auto@bib@innerbib\@empty
\bibitem [{\citenamefont {Dür}\ \emph {et~al.}(2001)\citenamefont {Dür},
  \citenamefont {Vidal}, \citenamefont {Cirac}, \citenamefont {Linden},\ and\
  \citenamefont {Popescu}}]{dur_entanglement_2001}%
  \BibitemOpen
  \bibfield  {author} {\bibinfo {author} {\bibfnamefont {W.}~\bibnamefont
  {Dür}}, \bibinfo {author} {\bibfnamefont {G.}~\bibnamefont {Vidal}},
  \bibinfo {author} {\bibfnamefont {J.~I.}\ \bibnamefont {Cirac}}, \bibinfo
  {author} {\bibfnamefont {N.}~\bibnamefont {Linden}},\ and\ \bibinfo {author}
  {\bibfnamefont {S.}~\bibnamefont {Popescu}},\ }\href
  {https://doi.org/10.1103/PhysRevLett.87.137901} {\bibfield  {journal}
  {\bibinfo  {journal} {Physical Review Letters}\ }\textbf {\bibinfo {volume}
  {87}},\ \bibinfo {pages} {137901} (\bibinfo {year} {2001})},\ \Eprint
  {https://arxiv.org/abs/quant-ph/0006034} {quant-ph/0006034} \BibitemShut
  {NoStop}%
\bibitem [{\citenamefont {Bennett}\ \emph {et~al.}(2003)\citenamefont
  {Bennett}, \citenamefont {Harrow}, \citenamefont {Leung},\ and\ \citenamefont
  {Smolin}}]{bennett_capacities_2003}%
  \BibitemOpen
  \bibfield  {author} {\bibinfo {author} {\bibfnamefont {C.}~\bibnamefont
  {Bennett}}, \bibinfo {author} {\bibfnamefont {A.}~\bibnamefont {Harrow}},
  \bibinfo {author} {\bibfnamefont {D.}~\bibnamefont {Leung}},\ and\ \bibinfo
  {author} {\bibfnamefont {J.}~\bibnamefont {Smolin}},\ }\href
  {https://doi.org/10.1109/TIT.2003.814935} {\bibfield  {journal} {\bibinfo
  {journal} {IEEE Transactions on Information Theory}\ }\textbf {\bibinfo
  {volume} {49}},\ \bibinfo {pages} {1895} (\bibinfo {year} {2003})},\ \Eprint
  {https://arxiv.org/abs/quant-ph/0205057} {quant-ph/0205057} \BibitemShut
  {NoStop}%
\bibitem [{\citenamefont {Bravyi}(2007)}]{bravyi_upper_2007}%
  \BibitemOpen
  \bibfield  {author} {\bibinfo {author} {\bibfnamefont {S.}~\bibnamefont
  {Bravyi}},\ }\href {https://doi.org/10.1103/PhysRevA.76.052319} {\bibfield
  {journal} {\bibinfo  {journal} {Physical Review A}\ }\textbf {\bibinfo
  {volume} {76}},\ \bibinfo {pages} {052319} (\bibinfo {year} {2007})},\
  \Eprint {https://arxiv.org/abs/0704.0964} {0704.0964} \BibitemShut {NoStop}%
\bibitem [{\citenamefont {Schlosshauer}(2005)}]{schlosshauer_decoherence_2005}%
  \BibitemOpen
  \bibfield  {author} {\bibinfo {author} {\bibfnamefont {M.}~\bibnamefont
  {Schlosshauer}},\ }\href {https://doi.org/10.1103/RevModPhys.76.1267}
  {\bibfield  {journal} {\bibinfo  {journal} {Reviews of Modern Physics}\
  }\textbf {\bibinfo {volume} {76}},\ \bibinfo {pages} {1267} (\bibinfo {year}
  {2005})},\ \Eprint {https://arxiv.org/abs/quant-ph/0312059}
  {quant-ph/0312059} \BibitemShut {NoStop}%
\bibitem [{\citenamefont {Van~Acoleyen}\ \emph {et~al.}(2013)\citenamefont
  {Van~Acoleyen}, \citenamefont {Mariën},\ and\ \citenamefont
  {Verstraete}}]{van_acoleyen_entanglement_2013}%
  \BibitemOpen
  \bibfield  {author} {\bibinfo {author} {\bibfnamefont {K.}~\bibnamefont
  {Van~Acoleyen}}, \bibinfo {author} {\bibfnamefont {M.}~\bibnamefont
  {Mariën}},\ and\ \bibinfo {author} {\bibfnamefont {F.}~\bibnamefont
  {Verstraete}},\ }\href {https://doi.org/10.1103/PhysRevLett.111.170501}
  {\bibfield  {journal} {\bibinfo  {journal} {Physical Review Letters}\
  }\textbf {\bibinfo {volume} {111}},\ \bibinfo {pages} {170501} (\bibinfo
  {year} {2013})},\ \Eprint {https://arxiv.org/abs/1304.5931} {1304.5931}
  \BibitemShut {NoStop}%
\bibitem [{\citenamefont {Mariën}\ \emph {et~al.}(2016)\citenamefont
  {Mariën}, \citenamefont {Audenaert}, \citenamefont {Van~Acoleyen},\ and\
  \citenamefont {Verstraete}}]{marien_entanglement_2016}%
  \BibitemOpen
  \bibfield  {author} {\bibinfo {author} {\bibfnamefont {M.}~\bibnamefont
  {Mariën}}, \bibinfo {author} {\bibfnamefont {K.~M.~R.}\ \bibnamefont
  {Audenaert}}, \bibinfo {author} {\bibfnamefont {K.}~\bibnamefont
  {Van~Acoleyen}},\ and\ \bibinfo {author} {\bibfnamefont {F.}~\bibnamefont
  {Verstraete}},\ }\href {https://doi.org/10.1007/s00220-016-2709-5} {\bibfield
   {journal} {\bibinfo  {journal} {Communications in Mathematical Physics}\
  }\textbf {\bibinfo {volume} {346}},\ \bibinfo {pages} {35} (\bibinfo {year}
  {2016})},\ \Eprint {https://arxiv.org/abs/1411.0680} {1411.0680} \BibitemShut
  {NoStop}%
\bibitem [{\citenamefont {Lieb}\ and\ \citenamefont
  {Vershynina}(2013)}]{lieb_upper_2013}%
  \BibitemOpen
  \bibfield  {author} {\bibinfo {author} {\bibfnamefont {E.~H.}\ \bibnamefont
  {Lieb}}\ and\ \bibinfo {author} {\bibfnamefont {A.}~\bibnamefont
  {Vershynina}},\ }\href {https://doi.org/10.26421/QIC13.11-12-5} {\bibfield
  {journal} {\bibinfo  {journal} {Quantum Information and Computation}\
  }\textbf {\bibinfo {volume} {13}},\ \bibinfo {pages} {986} (\bibinfo {year}
  {2013})},\ \Eprint {https://arxiv.org/abs/1302.3865} {1302.3865} \BibitemShut
  {NoStop}%
\bibitem [{\citenamefont {Audenaert}(2014)}]{audenaert_quantum_2014}%
  \BibitemOpen
  \bibfield  {author} {\bibinfo {author} {\bibfnamefont {K.~M.~R.}\
  \bibnamefont {Audenaert}},\ }\href {https://doi.org/10.1063/1.4901039}
  {\bibfield  {journal} {\bibinfo  {journal} {Journal of Mathematical Physics}\
  }\textbf {\bibinfo {volume} {55}},\ \bibinfo {pages} {112202} (\bibinfo
  {year} {2014})},\ \Eprint {https://arxiv.org/abs/1304.5935} {1304.5935}
  \BibitemShut {NoStop}%
\bibitem [{\citenamefont {Cheng}\ and\ \citenamefont
  {Liu}(2026)}]{cheng_error_2026}%
  \BibitemOpen
  \bibfield  {author} {\bibinfo {author} {\bibfnamefont {H.-C.}\ \bibnamefont
  {Cheng}}\ and\ \bibinfo {author} {\bibfnamefont {P.-C.}\ \bibnamefont
  {Liu}},\ }\href {https://doi.org/10.48550/arXiv.2507.06232} {\bibinfo {title}
  {Error {Exponents} for {Quantum} {Packing} {Problems} via {An} {Operator}
  {Layer} {Cake} {Theorem}}} (\bibinfo {year} {2026}),\ \Eprint
  {https://arxiv.org/abs/2507.06232} {2507.06232} \BibitemShut {NoStop}%
\bibitem [{\citenamefont {Frenkel}(2023)}]{frenkel_integral_2023}%
  \BibitemOpen
  \bibfield  {author} {\bibinfo {author} {\bibfnamefont {P.~E.}\ \bibnamefont
  {Frenkel}},\ }\href {https://doi.org/10.22331/q-2023-09-07-1102} {\bibfield
  {journal} {\bibinfo  {journal} {Quantum}\ }\textbf {\bibinfo {volume} {7}},\
  \bibinfo {pages} {1102} (\bibinfo {year} {2023})},\ \Eprint
  {https://arxiv.org/abs/2208.12194} {2208.12194} \BibitemShut {NoStop}%
\bibitem [{\citenamefont {Cheng}\ \emph {et~al.}(2025)\citenamefont {Cheng},
  \citenamefont {Gour}, \citenamefont {Lami},\ and\ \citenamefont
  {Liu}}]{cheng_operator_2025}%
  \BibitemOpen
  \bibfield  {author} {\bibinfo {author} {\bibfnamefont {H.-C.}\ \bibnamefont
  {Cheng}}, \bibinfo {author} {\bibfnamefont {G.}~\bibnamefont {Gour}},
  \bibinfo {author} {\bibfnamefont {L.}~\bibnamefont {Lami}},\ and\ \bibinfo
  {author} {\bibfnamefont {P.-C.}\ \bibnamefont {Liu}},\ }\href
  {https://doi.org/10.48550/arXiv.2512.04345} {\bibinfo {title} {The operator
  layer cake theorem is equivalent to {Frenkel}'s integral formula}} (\bibinfo
  {year} {2025}),\ \Eprint {https://arxiv.org/abs/2512.04345} {2512.04345}
  \BibitemShut {NoStop}%
\bibitem [{Note1()}]{Note1}%
  \BibitemOpen
  \bibinfo {note} {$t' = t\protect \tfrac {p}{q}$ with $p = \protect
  \operatorname {Tr}A$ and $q=1-p$.}\BibitemShut {Stop}%
\bibitem [{\citenamefont {Koßmann}\ \emph {et~al.}(2026)\citenamefont
  {Koßmann}, \citenamefont {Schwonnek}, \citenamefont {Liu},\ and\
  \citenamefont {Cheng}}]{kossmann_device-independent_2026}%
  \BibitemOpen
  \bibfield  {author} {\bibinfo {author} {\bibfnamefont {G.}~\bibnamefont
  {Koßmann}}, \bibinfo {author} {\bibfnamefont {R.}~\bibnamefont {Schwonnek}},
  \bibinfo {author} {\bibfnamefont {P.-C.}\ \bibnamefont {Liu}},\ and\ \bibinfo
  {author} {\bibfnamefont {H.-C.}\ \bibnamefont {Cheng}},\ }\href
  {https://doi.org/10.48550/arXiv.2607.03579} {\bibinfo {title}
  {Device-independent {Quantum} {Key} {Distribution} in the commuting operator
  framework}} (\bibinfo {year} {2026}),\ \Eprint
  {https://arxiv.org/abs/2607.03579} {2607.03579} \BibitemShut {NoStop}%
\bibitem [{\citenamefont {Mitrinović}(1970)}]{mitrinovic_analytic_1970}%
  \BibitemOpen
  \bibfield  {author} {\bibinfo {author} {\bibfnamefont {D.~S.}\ \bibnamefont
  {Mitrinović}},\ }\href {https://doi.org/10.1007/978-3-642-99970-3} {\emph
  {\bibinfo {title} {Analytic {Inequalities}}}}\ (\bibinfo  {publisher}
  {Springer Berlin Heidelberg},\ \bibinfo {address} {Berlin, Heidelberg},\
  \bibinfo {year} {1970})\BibitemShut {NoStop}%
\bibitem [{\citenamefont {Pimsner}\ and\ \citenamefont
  {Popa}(1986)}]{pimsner_entropy_1986}%
  \BibitemOpen
  \bibfield  {author} {\bibinfo {author} {\bibfnamefont {M.}~\bibnamefont
  {Pimsner}}\ and\ \bibinfo {author} {\bibfnamefont {S.}~\bibnamefont {Popa}},\
  }\href {https://doi.org/10.24033/asens.1504} {\bibfield  {journal} {\bibinfo
  {journal} {Annales scientifiques de l'École Normale Supérieure}\ }\textbf
  {\bibinfo {volume} {19}},\ \bibinfo {pages} {57} (\bibinfo {year}
  {1986})}\BibitemShut {NoStop}%
\bibitem [{\citenamefont {Audenaert}\ and\ \citenamefont
  {Kittaneh}(2012)}]{audenaert_problems_2012}%
  \BibitemOpen
  \bibfield  {author} {\bibinfo {author} {\bibfnamefont {K.~M.~R.}\
  \bibnamefont {Audenaert}}\ and\ \bibinfo {author} {\bibfnamefont
  {F.}~\bibnamefont {Kittaneh}},\ }\href
  {https://doi.org/10.48550/arXiv.1201.5232} {\bibinfo {title} {Problems and
  {Conjectures} in {Matrix} and {Operator} {Inequalities}}} (\bibinfo {year}
  {2012}),\ \Eprint {https://arxiv.org/abs/1201.5232} {1201.5232} \BibitemShut
  {NoStop}%
\bibitem [{\citenamefont {Vershynina}(2019)}]{vershynina_entanglement_2019}%
  \BibitemOpen
  \bibfield  {author} {\bibinfo {author} {\bibfnamefont {A.}~\bibnamefont
  {Vershynina}},\ }\href {https://doi.org/10.1063/1.5037802} {\bibfield
  {journal} {\bibinfo  {journal} {Journal of Mathematical Physics}\ }\textbf
  {\bibinfo {volume} {60}},\ \bibinfo {pages} {022201} (\bibinfo {year}
  {2019})},\ \Eprint {https://arxiv.org/abs/1803.07117} {1803.07117}
  \BibitemShut {NoStop}%
\bibitem [{\citenamefont {Luijk}\ and\ \citenamefont
  {Wilming}(2026)}]{luijk_sufficiency_2026}%
  \BibitemOpen
  \bibfield  {author} {\bibinfo {author} {\bibfnamefont {L.~v.}\ \bibnamefont
  {Luijk}}\ and\ \bibinfo {author} {\bibfnamefont {H.}~\bibnamefont
  {Wilming}},\ }\href {https://doi.org/10.48550/arXiv.2604.08380} {\bibinfo
  {title} {Sufficiency and {Petz} recovery for positive maps}} (\bibinfo {year}
  {2026}),\ \Eprint {https://arxiv.org/abs/2604.08380} {2604.08380}
  \BibitemShut {NoStop}%
\bibitem [{\citenamefont {Silva}\ \emph {et~al.}(2026)\citenamefont {Silva},
  \citenamefont {Fröb}, \citenamefont {Lechner},\ and\ \citenamefont
  {Sangaletti}}]{silva_integral_2026}%
  \BibitemOpen
  \bibfield  {author} {\bibinfo {author} {\bibfnamefont {R.~C.~d.}\
  \bibnamefont {Silva}}, \bibinfo {author} {\bibfnamefont {M.~B.}\ \bibnamefont
  {Fröb}}, \bibinfo {author} {\bibfnamefont {G.}~\bibnamefont {Lechner}},\
  and\ \bibinfo {author} {\bibfnamefont {L.}~\bibnamefont {Sangaletti}},\
  }\href {https://doi.org/10.48550/arXiv.2607.05195} {\bibinfo {title}
  {Integral representations of $f$-divergences for general von {Neumann}
  algebras}} (\bibinfo {year} {2026}),\ \Eprint
  {https://arxiv.org/abs/2607.05195} {2607.05195} \BibitemShut {NoStop}%
\bibitem [{\citenamefont {Araki}(1976)}]{araki_relative1_1976}%
  \BibitemOpen
  \bibfield  {author} {\bibinfo {author} {\bibfnamefont {H.}~\bibnamefont
  {Araki}},\ }\href {https://doi.org/10.2977/prims/1195191148} {\bibfield
  {journal} {\bibinfo  {journal} {Publications of the Research Institute for
  Mathematical Sciences}\ }\textbf {\bibinfo {volume} {11}},\ \bibinfo {pages}
  {809} (\bibinfo {year} {1976})}\BibitemShut {NoStop}%
\bibitem [{\citenamefont {Araki}(1977)}]{araki_relative2_1977}%
  \BibitemOpen
  \bibfield  {author} {\bibinfo {author} {\bibfnamefont {H.}~\bibnamefont
  {Araki}},\ }\href {https://doi.org/10.2977/prims/1195190105} {\bibfield
  {journal} {\bibinfo  {journal} {Publications of the Research Institute for
  Mathematical Sciences}\ }\textbf {\bibinfo {volume} {13}},\ \bibinfo {pages}
  {173} (\bibinfo {year} {1977})}\BibitemShut {NoStop}%
\bibitem [{\citenamefont {Kosaki}(1986)}]{kosaki_relative_1986}%
  \BibitemOpen
  \bibfield  {author} {\bibinfo {author} {\bibfnamefont {H.}~\bibnamefont
  {Kosaki}},\ }\href {https://www.jstor.org/stable/24714805} {\bibfield
  {journal} {\bibinfo  {journal} {Journal of Operator Theory}\ }\textbf
  {\bibinfo {volume} {16}},\ \bibinfo {pages} {335} (\bibinfo {year}
  {1986})}\BibitemShut {NoStop}%
\bibitem [{\citenamefont {Ohya}\ and\ \citenamefont
  {Petz}(1993)}]{ohya_quantum_1993}%
  \BibitemOpen
  \bibfield  {author} {\bibinfo {author} {\bibfnamefont {M.}~\bibnamefont
  {Ohya}}\ and\ \bibinfo {author} {\bibfnamefont {D.}~\bibnamefont {Petz}},\
  }\href {https://link.springer.com/book/9783540208068} {\emph {\bibinfo
  {title} {Quantum {Entropy} and {Its} {Use}}}},\ Theoretical and
  {Mathematical} {Physics}\ (\bibinfo  {publisher} {Springer Berlin,
  Heidelberg},\ \bibinfo {year} {1993})\BibitemShut {NoStop}%
\bibitem [{\citenamefont {Lieb}(1973)}]{lieb_convex_1973}%
  \BibitemOpen
  \bibfield  {author} {\bibinfo {author} {\bibfnamefont {E.~H.}\ \bibnamefont
  {Lieb}},\ }\href {https://doi.org/10.1016/0001-8708(73)90011-X} {\bibfield
  {journal} {\bibinfo  {journal} {Advances in Mathematics}\ }\textbf {\bibinfo
  {volume} {11}},\ \bibinfo {pages} {267} (\bibinfo {year} {1973})}\BibitemShut
  {NoStop}%
\bibitem [{\citenamefont {Robertson}(1929)}]{robertson_uncertainty_1929}%
  \BibitemOpen
  \bibfield  {author} {\bibinfo {author} {\bibfnamefont {H.~P.}\ \bibnamefont
  {Robertson}},\ }\href {https://doi.org/10.1103/PhysRev.34.163} {\bibfield
  {journal} {\bibinfo  {journal} {Physical Review}\ }\textbf {\bibinfo {volume}
  {34}},\ \bibinfo {pages} {163} (\bibinfo {year} {1929})}\BibitemShut
  {NoStop}%
\bibitem [{Note2()}]{Note2}%
  \BibitemOpen
  \bibinfo {note} {The inclusions are understood as extending a compact
  operator on a closed subspace by zero on the complement.}\BibitemShut {Stop}%
\bibitem [{\citenamefont {van Luijk}\ \emph {et~al.}(2024)\citenamefont {van
  Luijk}, \citenamefont {Stottmeister},\ and\ \citenamefont
  {Werner}}]{van_luijk_convergence_2024}%
  \BibitemOpen
  \bibfield  {author} {\bibinfo {author} {\bibfnamefont {L.}~\bibnamefont {van
  Luijk}}, \bibinfo {author} {\bibfnamefont {A.}~\bibnamefont {Stottmeister}},\
  and\ \bibinfo {author} {\bibfnamefont {R.~F.}\ \bibnamefont {Werner}},\
  }\bibfield  {journal} {\bibinfo  {journal} {Annales Henri Poincaré}\ }\href
  {https://doi.org/10.1007/s00023-024-01413-6} {10.1007/s00023-024-01413-6}
  (\bibinfo {year} {2024}),\ \Eprint {https://arxiv.org/abs/2306.16063}
  {2306.16063} \BibitemShut {NoStop}%
\bibitem [{\citenamefont {Blackadar}(2006)}]{blackadar_operator_2006}%
  \BibitemOpen
  \bibfield  {author} {\bibinfo {author} {\bibfnamefont {B.}~\bibnamefont
  {Blackadar}},\ }\href {https://doi.org/10.1007/3-540-28517-2} {\emph
  {\bibinfo {title} {Operator {Algebras}}}},\ edited by\ \bibinfo {editor}
  {\bibfnamefont {J.}~\bibnamefont {Cuntz}}\ and\ \bibinfo {editor}
  {\bibfnamefont {V.}~\bibnamefont {F.R.~Jones}},\ \bibinfo {series}
  {Encyclopaedia of {Mathematical} {Sciences}}, Vol.\ \bibinfo {volume} {122}\
  (\bibinfo  {publisher} {Springer},\ \bibinfo {address} {Berlin, Heidelberg},\
  \bibinfo {year} {2006})\BibitemShut {NoStop}%
\bibitem [{\citenamefont {Reed}\ and\ \citenamefont
  {Simon}(1980)}]{reed_functional_1980}%
  \BibitemOpen
  \bibfield  {author} {\bibinfo {author} {\bibfnamefont {M.}~\bibnamefont
  {Reed}}\ and\ \bibinfo {author} {\bibfnamefont {B.}~\bibnamefont {Simon}},\
  }\href@noop {} {\emph {\bibinfo {title} {Functional analysis}}},\ \bibinfo
  {series} {Methods of modern mathematical physics}\ No.~\bibinfo {number} {1}\
  (\bibinfo  {publisher} {Academic Press},\ \bibinfo {address} {San Diego,
  California},\ \bibinfo {year} {1980})\BibitemShut {NoStop}%
\end{thebibliography}

\clearpage

\onecolumngrid
\bigskip
\section{Appendix}
\twocolumngrid

\paragraph{Justification of Eqs.~\eqref{eq:log_com}, \eqref{eq:log_int} \& \eqref{eq:log_rep}.}

Following \cite[Eq.~(3.6)]{lieb_convex_1973}, the directional derivative of the operator logarithm on the left-hand side of \eqref{eq:kernel} can be written as the norm-convergent integral
\begin{align}\label{eq:kernel_alt}
  \Dlog[q\sigma\!+\!s\rho](\rho)& = \!\!\int_{0}^{\infty}\!\!\!dr\,\frac{1}{q\sigma\!+\!s\rho\!+\!r\1}\rho\frac{1}{q\sigma\!+\!s\rho\!+\!r\1}.
\end{align}
For every $s\in[0,p]$ and $r\geq0$, it follows that $q\sigma+s\rho+r\1\geq(m+r)\1$, where $m=q\lambda_{\min}(\sigma)>0$. Thus, using the representation \eqref{eq:kernel_alt} for $\Dlog[q\sigma+s\rho](\rho)$, results in an integrand that is norm continuous in $s\in[0,p]$ and is dominated by $\norm{\rho}(m+r)^{-2}$. Therefore, \eqref{eq:kernel} is norm continuous in $s$, and the expression \eqref{eq:log_int} is justified.

This results in
\begin{align}
  \log B-\log(q\sigma) & = \int_{0}^{p} ds\, \int_{0}^{\infty} dt\,\tfrac{q}{(q+st)^2}\,P_t.
\end{align}
Switching the integration order is justified because $(s,t)\mapsto q(q+st)^{-2}P_t$ is Borel measurable, bounded by $1/q$, and supported in the compact rectangle $[0,p]\times[0,R]$ with $R=\norm{\sigma^{-1/2}\rho\sigma^{-1/2}}$ ($\rho\leq R\sigma$ gives $P_t=0$ for $t\geq R$). Thus, each matrix entry is absolutely integrable, and Fubini's theorem applies entrywise, yielding \eqref{eq:log_rep}.

Finally, the map $T\mapsto[\tau,T]$ from $\Bcal(\Hcal)$ to $\Sone(\Hcal)$ is bounded, $\tno{[\tau,T]}\leq2\tno{\tau}\norm{T}$, and, thus, commutes with the integral in \eqref{eq:log_rep}, which yields the desired expression \eqref{eq:log_com}.

\paragraph{The proof of Lemma \ref{lem:proj}.}

For positive $\tau\in\Sone(\Hcal)$ with $\Tr\tau\!=\!1$,
\begin{align}
\label{eq:tau_ip}
\langle T,S\rangle_{\tau} & = \Tr(\tau T^{*}S), & T,S & \in\Bcal(\Hcal),
\end{align}
defines a positive sesquilinear form on $\Bcal(\Hcal)$. For bounded $T=T^{*}$, it follows that (by a variation of Robertson's uncertainty principle \cite{robertson_uncertainty_1929})
\begin{align}
\label{eq:tno_com_bound}
\tno{[\tau,T]} & \!=\! \sup_{-\1\leq S \leq\1}\!\!|\Tr([\tau,T]S)| \!=\! \sup_{-\1\leq S \leq\1}\!\!|\Tr(\tau[T,S])| \notag \\
& \!=\! 2\sup_{-\1\leq S \leq\1}\!\!|\Im\langle T,S\rangle_{\tau}| \!\leq\! 2\langle T^{2}\rangle_{\tau}^{\frac{1}{2}}\sup_{-\1\leq S \leq\1}\!\!\langle S^{2}\rangle_{\tau}^{\frac{1}{2}} \notag \\
& \!\leq\! 2\langle T^{2}\rangle_{\tau}^{\frac{1}{2}},
\end{align}
using the variational characterization of the trace norm for skew-adjoint operators, the Cauchy--Schwarz inequality for \eqref{eq:tau_ip}, and $\langle S^{2}\rangle_{\tau}^{\frac{1}{2}}\!\leq\!\norm{S}$. Replacing $T$ by $T\!-\!\langle T\rangle_{\tau}$ implies the first inequality of \eqref{eq:trace_bound}.

Assuming in addition that $T$ is positive, the final inequality in \eqref{eq:trace_bound} is implied by the operator inequality $T\leq\norm{T}$ and the fact that $S\mapsto T^{\frac{1}{2}}ST^{\frac{1}{2}}$ is completely positive:
\begin{align}
\langle T^{2}\rangle_{\tau}\!-\!\langle T\rangle_{\tau}^{2} & \!\leq\! \norm{T}\langle T\rangle_{\tau}\!-\!\langle T\rangle_{\tau}^{2} \notag \\
& \!=\! \tfrac{1}{4}\norm{T}^{2}\!-\!(\tfrac{1}{2}\norm{T}\!-\!\langle T\rangle_{\tau})^2 \notag \\
& \!\leq\! \tfrac{1}{4}\norm{T}^{2}.
\end{align}
This implies the second inequality in \eqref{eq:trace_bound}.

\paragraph{The proof of Theorem~\ref{thm:main} for separable Hilbert spaces.}

Now, let $\dim\Hcal_{0}=\infty$, and let $(e_{j})_{j\geq1}$ be an orthonormal basis of eigenvectors of $B$,
\begin{align}\label{eq:bev}
Be_{j} & =b_{j}\, e_{j},
\end{align}
with $b_{j}>0$, $\sum_{j}b_{j}=1$. Let $P_{n}$ be the projection onto the closed subspace $\Hcal_{n}\subseteq\Hcal_{0}$ spanned by $\{e_{1},...,e_{n}\}$ and set $A_{n}=P_{n}AP_{n}$, $B_{n}=P_{n}BP_{n}$, $\alpha_{n}=\Tr A_{n}\uparrow p$, $\beta_{n}=\Tr B_{n}\uparrow1$. Applying the finite-dimensional case to the normalized pair $(\tfrac{1}{\beta_{n}}A_n,\tfrac{1}{\beta_{n}}B_n)$ on $\Hcal_{n}$ gives, for $D_{n}:=[A_{n},\log B_{n}]$,
\begin{align}\label{eq:Dnbound}
  \tno{D_{n}} & \leq\beta_{n}\,\hbin(\tfrac{\alpha_{n}}{\beta_{n}})\to\hbin(p).
\end{align}
Clearly, $P_{m}P_{n}=P_{n}P_{m}$ as well as $P_{m}\log B_{n}P_{m} = \log B_{m}$ for $m\leq n$ and, thus, $P_{m}D_{n}P_{m}=D_{m}$. Now, $\varphi_{m}(K_{m}):=\Tr(K_{m}D_m)$ with $K_{m}\in\Kcal(\Hcal_{m})$ defines a consistent sequence of linear functionals on the nested sequence $...\subseteq\Kcal(\Hcal_{n})\subseteq\Kcal(\Hcal_{n+1})\subseteq...\subseteq\Kcal(\Hcal_{0})$.\footnote{The inclusions are understood as extending a compact operator on a closed subspace by zero on the complement.} Since any element $K\in\Kcal(\Hcal_{0})$ can be norm-approximated by its compressions $K_{n}=P_{n}KP_{n}$, it follows that $|\varphi_{n}(K_{n})|\leq\tno{D_{n}}\norm{K_{n}}\leq\beta_{n}\,\hbin(\tfrac{\alpha_{n}}{\beta_{n}})\norm{K_{n}}\to\hbin(p)\norm{K}$ by \eqref{eq:Dnbound} and, thus, defines a unique bounded linear functional $\varphi$ on $\Kcal(\Hcal_{0})$ \cite{van_luijk_convergence_2024}.
By the duality between compact and trace-class operators, $\Kcal(\Hcal_{0})^*=\Sone(\Hcal_{0})$ \cite[Thm.~I.8.6.1]{blackadar_operator_2006}, $\varphi$ has a unique representative $D\in\Sone(\Hcal_{0})$ with $\tno{D}\leq\hbin(p)$ and $P_{n}DP_{n}=D_{n}$. Again, since compressions of a fixed trace-class operator converge in trace norm, it follows that $\tno{D_{n}-D}\to0$. Finally, the matrix elements
\begin{align}
  \langle e_{i},De_{j}\rangle & =(\log b_{j}-\log b_{i})\,\langle e_{i},Ae_{j}\rangle
\end{align}
show that $D$ represents the form $\mathfrak c$ in \eqref{eq:weakform} on $\mathcal D_{0} = \textup{span}\{e_{j}\}$, which is a core for $\log B$ in the graph norm. The identity extends to all of $\operatorname{Dom}(\log B)$ by graph continuity, and density of $\mathcal D_{0}$ in $\Hcal_{0}$ makes the bounded representative unique \cite[Chp.~VIII]{reed_functional_1980}. This completes the proof of Theorem~\ref{thm:main}.

\end{document}